\documentclass[letterpaper, 10 pt, conference]{ieeeconf}  

\IEEEoverridecommandlockouts                              

\usepackage{cite}
\usepackage{amsmath,amssymb,amsfonts}
\usepackage{commath}
\usepackage{bm}
\usepackage{algpseudocode}
\usepackage{algorithm}
\usepackage{graphicx}
\usepackage{subcaption}
\usepackage{textcomp}
\usepackage{xcolor}
\usepackage{multirow}
\usepackage{float}
\usepackage{booktabs}
\usepackage{stackengine}
\usepackage{etoolbox}
\usepackage{mathtools}
\usepackage{subcaption}
\usepackage{optidef}

\newtheorem{assumption}{Assumption}

\newtheorem{theorem}{Theorem}

\newtheorem{corollary}{Corollary}
\newtheorem{remark}{Remark}

\newcommand{\longdash}{-}
\usepackage[utf8]{inputenc}
\DeclareUnicodeCharacter{2212}{-}
\newcommand{\RNum}[1]{\uppercase\expandafter{\romannumeral#1\relax}}
\DeclareMathOperator{\sign}{sign}

\title{\LARGE \bf
	Robust Control Design Using a Hybrid-Gain Finite-Time Sliding-Mode Controller
}

\author{Amit Shivam$^{1}$, Kiran Kumari$^{2}$, and Fernando A.C.C. Fontes$^{1}$ 
	\thanks{$^{1}$Amit Shivam and Fernando A.C.C. Fontes are with the Department of Electrical and Computer Engineering, University of Porto, Porto, Portugal, Email:
		{\tt\small amit.shivam1407@gmail.com}, and {\tt\small faf@fe.up.pt}}
	\thanks{$^{2}$Kiran Kumari is with Department of Electrical Engineering, Indian Institute of Science, Bengaluru, India, Email: {\tt\small kirank@iisc.ac.in}}
}

\begin{document}
	
	\maketitle
	\thispagestyle{empty}
	\pagestyle{empty}


\begin{abstract}
This paper proposes a hybrid-gain finite-time sliding-mode control (HG-FTSMC) strategy for a class of perturbed nonlinear systems. 
The controller combines a finite-time reaching law that drives the sliding variable to a predefined boundary layer with an inner mixed-power or exponential law that guarantees rapid convergence within the layer while maintaining smooth and bounded control action. 
The resulting control design achieves finite-time convergence and robustness to matched disturbances, while explicitly limits the control effort.
The control framework is first analyzed on a perturbed first-order integrator model, and then extended to Euler-Lagrange (EL) systems, representing a broad class of robotic and mechanical systems. 
Comparative simulations demonstrate that the proposed controller achieves settling times comparable to recent finite-time approaches~\cite{li2021simultaneous}, while substantially reducing the control effort. 
Finally, trajectory-tracking simulations on a two-link manipulator further validate the robustness and practical feasibility of the proposed HG-FTSMC approach.
\end{abstract}

				
	\begin{keywords}
		Finite and Fixed-time sliding mode control, matched disturbances, robust control
	\end{keywords}

    \section{Introduction}

Finite-time sliding-mode control (FT-SMC) has attracted significant attention due to its ability to enforce rapid, precise, and robust convergence of system states within a strictly finite interval.
Unlike classical asymptotic control methods, FT-SMC explicitly guarantees bounded settling time, thereby offering superior transient performance and strong resilience against matched disturbances and uncertainties \cite{ds1998continuous,feng2002non,feng2013nonsingular,basin2019finite,song2023prescribed}.
Such properties make FT-SMC well suited for safety-critical and high-performance application, including robotic manipulation \cite{galicki2015finite},  spacecraft attitude regulation \cite{du2011finite}, and cooperative multi-agent coordination \cite{li2017finite}, where responsiveness and disturbance rejection are of paramount importance.

Despite these advantages, classical FT-SMC formulations suffer from two fundamental limitations that restrict their feasibility to real-world platforms, such as unmanned aerial vehicles (UAVs), ground robots, and other actuator-constrained systems.
First, many FT-SMC designs rely on nonlinear gains that grow rapidly with the tracking error, such as exponential or reciprocal-power terms \cite{hu2024modified,shang2025pmsm}.
Although such gains theoretically accelerate convergence, they usually generate excessively large control signals during the initial transient, particularly when the system begins far from the sliding surface.
For physical systems with strict torque, thrust, or steering limits, such large control demands can lead to actuator saturation, degraded tracking accuracy, and even instability under bounded disturbances. 
Second, although several improved finite-time or fixed-time (FxT) approaches~\cite{polyakov2011nonlinear,zuo2015non,moulay2021robust,labbadi2023design} attempt to moderate this behavior, most retain fixed nonlinear exponents or gain structures that still produce large transient peaks when initial errors are significant, reducing their practical utility.

Norm-normalized (NN) discontinuous functions have been developed to address chattering and increase smoothness in sliding-mode designs~\cite{li2021simultaneous,li2021time}. 
By scaling the discontinuous term using the state norm, NN-based controllers reduce high-frequency switching near the origin and achieve simultaneous convergence of all states (SATO property). 
However, this normalization inherently limit the control amplitude to the magnitude of the state vector. 
Consequently, NN controllers may increase the required actuation for large initial errors, while attenuating the action when the system is far from the sliding manifold $s$, which is exactly opposite of what is desired for robust and fast transient response. 
Additionally, because the normalization depends on $\|s\|$, NN methods may exhibit sensitivity or pseudo-singular behavior as the state norm approaches zero, making their deployment complicated in realistic applications. 
Thus, while NN-based approaches mitigate chattering, they do not fully resolve the need for bounded, actuator-safe finite-time controllers with predictable transient behavior.


Motivated by these limitations, this work develops a HG-FTSMC strategy that explicitly incorporates actuator limitations while preserving fast and robust finite-time convergence.
The key idea is to design a two-region state-dependent gain:
\begin{itemize}
    \item \emph{an outer finite-time region}, where a bounded mixed-power gain accelerates convergence without producing large control peaks, even under large initial errors; and
\item an \textit{inner fixed-time region} employing a smooth polynomial/exponential gain to ensure state-independent convergence once trajectories approach a predefined neighborhood of the sliding manifold.
\end{itemize}
This hybrid mechanism enables flexible shaping of the transient behavior, ensures actuator-safe control effort across the entire state space, and eliminates the need for normalization unlike NN-based schemes.
Moreover, the design maintains robustness against bounded disturbances while ensuring smooth transition between the two regions.

The main contributions of this paper are as follows:
\begin{enumerate}
    \item A new hybrid-gain finite-/fixed-time SMC framework that integrates a bounded outer finite-time gain with an inner fixed-time gain, providing actuator-feasible control effort and smooth yet rapid convergence.
    \item A rigorous Lyapunov stability analysis establishing finite-time reaching into a predefined boundary layer and fixed-time convergence within the boundary layer, with explicit closed-form expressions for both settling intervals.
    \item A comparative study with NN-based SMC approach (SATO) ~\cite{li2021simultaneous}, demonstrats that the proposed design achieves comparable or faster convergence with significantly reduced transient control peaks and without singularity issues.
    \item Simulations results on a perturbed first-order integrator model and a two-link manipulator, validating the robustness and practical implementability of the proposed HG-FTSMC controller.
\end{enumerate}

    The remainder of this paper is organized as follows. The problem formulation is discussed
in Section~\ref{sec:problem}. In Section~\ref{sec: proposed controller design}, the proposed control scheme
will be derived and its stability will be analyzed, followed
by performance validation using the numerical simulation
results in Section~\ref{sec: simulation results}. Finally, Section~\ref{sec:conclusion} concludes the paper.
 
\section{Problem Formulation}
\label{sec:problem}

We consider nonlinear systems of the form
\begin{equation}
\dot{x}(t)=f(x(t),u(t)) + d(t), 
\label{eq:general}
\end{equation}
where $x\in\mathbb{R}^n$, $u\in\mathbb{R}^m$ is the control input, and 
$d(t)\in\mathbb{R}^n$ is a matched disturbance.

\begin{assumption}
\label{ass:dist}
The disturbance $d(t)$ is bounded componentwise as
\begin{equation}
|d_i(t)| \le \bar d_i,\qquad \bar d_i>0.
\end{equation}
\end{assumption}

\noindent
\textit{Control objective.}
Given any initial condition $x_0$, design a bounded control law $u(t)$ such that the closed-loop trajectories satisfy
\[
x(t)\to 0
\]
in finite or fixed time, while ensuring robustness with respect to any disturbance satisfying Assumption~\ref{ass:dist}, and guaranteeing uniformly bounded control effort.

\vspace{0.8em}
To illustrate and validate the proposed hybrid-gain finite-time sliding-mode control framework, two representative subclasses of~\eqref{eq:general} are considered.

\paragraph*{1) First-order perturbed integrator}
\begin{equation}
\dot{x}(t)=u(t)+d(t),
\label{eq:first-order}
\end{equation}
which captures the dominant dynamics of sliding variables, and serves as the canonical setting for establishing finite-/fixed-time convergence and robustness properties.

\paragraph*{2) Euler--Lagrange (EL) system}
\begin{equation}
M(q)\ddot q + C(q,\dot q)\dot q + g(q)=\tau + d(t),
\label{eq:EL}
\end{equation}
where $q\in\mathbb{R}^n$,
$M(q)$ is symmetric positive definite,
$C(q,\dot q)$ satisfies the standard EL skew-symmetry property,
and $\tau\in\mathbb{R}^n$ is the control input.
The disturbance $d(t)$ is matched as in Assumption~\ref{ass:dist}.
Model~\eqref{eq:EL} encompasses a broad class of robotic and mechanical systems and enables the extension of the proposed method to multi-DOF dynamics~\cite{du2011finite,galicki2015finite,li2017finite}.

\section{Proposed Controller Design}\label{sec: proposed controller design}
We now introduce a hybrid gain–based sliding-mode control law that eliminates the exponential growth in control effort~\cite{labbadi2023design}, particularly under large initial errors, while still guaranteeing finite-time (or practical fixed-time) convergence in the presence of bounded disturbances. This formulation provides a smoother alternative and its settling-time behavior relative to the SATO framework is discussed next.
\subsection{Intuitive Settling-Time Estimate for the SATO Controller}
Following~\cite{li2021simultaneous}, let $s\in\mathbb{R}^n$ be the sliding vector and define
the infinity norm
\begin{equation}
r(t) := \|s(t)\|_\infty = \max_i |s_i(t)|.
\end{equation}
Consider the nominal SATO dynamics
\begin{equation}
\dot s = -K\,\frac{s}{\|s\|_2}+d(t), \qquad \|d(t)\|_\infty \le \bar d,\quad K>0,
\label{eq:sato-dyn}
\end{equation}
and let $V(s)=r(t)$.
Using the upper-right derivative $D^+V$ and standard norm inequalities,
one obtains
\begin{equation}
D^+ r \;\le\; -K + \bar d \;=\; -\big(K-\bar d\big).
\label{eq:sato-dr}
\end{equation}
Thus $r(t)$ decreases at a constant rate.  
For any region with radius $\varepsilon>0$, the time required for the SATO controller
to enter the set $\{\,\|s\|_\infty\le\varepsilon\,\}$ satisfies
\begin{equation}
T_{\mathrm{SATO}}
\;\le\;
\frac{r(0)-\varepsilon}{K-\bar d}.
\label{eq:sato-T}
\end{equation}

\begin{remark}[Interpretation of the SATO settling time]
The bound~\eqref{eq:sato-T} shows that the norm $\|s\|_\infty$ decreases
\emph{linearly} with time, with slope $K-\bar d$.
Hence:
\begin{itemize}
  \item The convergence time grows linearly with the initial error magnitude $r(0)$; large initial conditions require proportionally longer time to reach the boundary layer.
  \item The control magnitude is essentially constant and equal to $K$ away from the origin (norm-normalized sign), so increasing $K$ to accelerate convergence directly increases the applied control effort.
  \item The SATO design guarantees simultaneous arrival (ratio persistence) in the disturbance-free case, but does not provide a fixed-time bound: $T_{\mathrm{SATO}}$ is not uniform with respect to $r(0)$.
\end{itemize}
\end{remark}

\begin{remark}[Limitations and motivation for proposed design]
The formulation~\eqref{eq:sato-dyn}--\eqref{eq:sato-T} highlights a structural limitation
of the SATO approach: the decay rate of $\|s\|_\infty$ is constant and entirely determined
by the single gain $K$, leading to a trade-off between fast convergence and large,
uniform control effort. In particular, for very large initial errors one must choose
a large $K$ to reduce $T_{\mathrm{SATO}}$, which in turn yields high-amplitude control
action over the entire transient.

In contrast, the proposed hybrid-gain controller discussed subsequently in Theorem~\ref{thm:firstorder-hybrid}
replaces the constant gain $K$ by a state-dependent gain
\[
G_{\mathrm{hyb}}(x) =
\begin{cases}
G_{\mathrm{out}}(x)=k_0+k_1\dfrac{|x|^\gamma}{\varepsilon_0^\gamma+|x|^\gamma},
   & |x|>\varepsilon,\\[1.2ex]
G_{\mathrm{in}}(x)=a|x|^\gamma+b|x|^\alpha,
   & |x|\le\varepsilon,
\end{cases}
\]
with $k_0>\bar d$, $k_1>0$, $0<\gamma<1<\alpha$.
For the outer region, the Lyapunov analysis leads to
\[
D^+|x|
\;\le\;
-\big(k_0-\bar d\big)
- k_1\,\frac{|x|^\gamma}{\varepsilon_0^\gamma+|x|^\gamma},
\]
which yields an explicit boundary layer-entry time $T_{\mathrm{out}}$ strictly smaller
than~\eqref{eq:sato-T} for the same robustness margin ($k_0\simeq K$), while keeping
the outer control uniformly bounded by $k_0+k_1$.
Inside the boundary layer, the mixed-power inner law $G_{\mathrm{in}}(x)$ guarantees a
fixed-time bound independent of the initial condition.
Thus the overall HG--FTSMC design overcomes the constant-rate limitation of
the SATO formulation, providing faster convergence for large errors and a
state-independent completion time, with bounded control effort.
\end{remark}
\subsection{First-Order System}
\begin{theorem}
\label{thm:firstorder-hybrid}
Given the system dynamics~\eqref{eq:first-order}, the proposed control law
\begin{equation}
u(x)\;=\; -\,G_{\mathrm{hyb}}(x)\,\mathrm{sgn}(x),
\label{eq: first order control law}
\end{equation}
and, hybrid gain 
\begin{align}
  G_{\mathrm{hyb}}(x) \;=\;
\begin{cases}
G_{\mathrm{out}}(x), & |x|>\varepsilon,\\[0.3ex]
G_{\mathrm{in}}(x),  & |x|\le \varepsilon,
\end{cases}
\label{eq: hybrid gain}
\end{align}
drives the system into equilibrium (at origin).
 For any $\varepsilon>0$ such that $\varepsilon \in (0,\varepsilon_o]$, the variable $x(t) $  enters the boundary layer 
$|x|\le\varepsilon$ in finite time, employing the hybrid gain  proposed as
\begin{align}
\begin{split}
G_{\mathrm{out}}(x)
&= k_0 + k_1\,\frac{|x|^{\gamma}} {\varepsilon_0^{\gamma}+|x|^{\gamma}}, \\
& \text{for} \quad k_0>\bar d,\; k_1>0,\; \varepsilon_0>0,\; 0<\gamma<1. 
\label{eq:Gout}
\end{split}
\end{align}
\begin{align}
    \begin{split}
 G_{\mathrm{in}}(x)
&= a\,|x|^{\gamma} + b\,|x|^{\alpha}, \\
& \text{for} \quad a>0,\; b>0,\; 0<\gamma<1<\alpha.
\label{eq:Gin-poly}       
    \end{split}
\end{align}
 The convergence time follows the bound 
\begin{enumerate}
\item[(i)] \textit{FT reaching to the boundary layer.} 
From any $x(0)=x_0$, the state enters the boundary layer $|x|\le \varepsilon$ in finite time $T_{\mathrm{out}}$ bounded by
\begin{align}
T_{\mathrm{out}}
&\le \frac{1}{k_0-\bar d}\,
\ln\!\frac{\max\{|x_0|,\varepsilon_0\}}{\max\{\varepsilon,\varepsilon_0\}}
\notag\\
&\hspace{1.0ex}
+ \frac{\varepsilon_0^{\gamma}}{k_1(1-\gamma)}
\Big(\varepsilon^{\,1-\gamma}-\min\{|x_0|,\varepsilon_0\}^{\,1-\gamma}\Big)_+ .
\label{eq:Tout-bound}
\end{align}
\item[(ii)] \textit{FxT convergence inside the boundary layer.}
Once $|x|\le \varepsilon$, the origin is reached within a \emph{state-independent} time
\begin{equation}
T_{\mathrm{in}}
\;\le\; \frac{1}{a(1-\gamma)} \;+\; \frac{1}{b(\alpha-1)} .
\label{eq:Tin-poly}
\end{equation}
The overall convergence time satisfies,  Total settling time $T_\mathrm{tot}$
\begin{equation}
T_{\mathrm{tot}} \;\le\; T_{\mathrm{out}} + T_{\mathrm{in}}.
\end{equation}
\end{enumerate}
\end{theorem}

\begin{proof}By choosing Lyapunov function as $V=|x|$ and on taking its derivative, one can obtain
\begin{align}
\dot V = \dot x \sign (x) +d(t)\sign (x).
\label{eq: Vdot first order}
\end{align}
Using  ~\eqref{eq: first order control law} and  ~\eqref{eq: Vdot first order}
\begin{align}
    \dot V = -G_{\mathrm{hyb}}(x)+d(t) \sign (x).
    \label{eq: Vdot first order derived}
\end{align}
\noindent\textbf{Case 1:} $|x| > \varepsilon$.\\[0.3em]
 Using  ~\eqref{eq:Gout} and \eqref{eq: Vdot first order derived} results in
\begin{align}
 \dot V \le -(k_0-\bar d) \;- k_1\frac{V^\gamma}{\varepsilon_0^\gamma+V^\gamma} 
 \label{eq: Vdot derived total}
\end{align}


Since $k_0>\bar d$, the right--hand side is strictly negative for all $V>0$, hence $V(t)$ is strictly decreasing.
To obtain an explicit time bound to reach $\varepsilon$, we split the time evolution into two ranges.

\medskip
\noindent\emph{Range 1 ($V\ge \varepsilon_0$):} 
Here $\displaystyle \frac{V^{\gamma}}{\varepsilon_0^{\gamma}+V^{\gamma}} \ge \frac{1}{2}$.\newline 
From \eqref{eq: Vdot derived total},
\begin{equation}
\dot V \le -(k_0-\bar d) - \tfrac{k_1}{2}
\;\le\; -(k_0-\bar d).
\end{equation}
Integrating $\dot V \le -(k_0-\bar d)$ from $V(0)=\max\{|x_0|,\varepsilon_0\}$ down to $V=\max\{\varepsilon,\varepsilon_0\}$
gives
\begin{equation}
t_1 \;\le\; \frac{1}{k_0-\bar d}\,
\ln\!\frac{\max\{|x_0|,\varepsilon_0\}}{\max\{\varepsilon,\varepsilon_0\}}.
\label{eq:t1-range1}
\end{equation}

\medskip
\noindent\emph{Range 2 ($\varepsilon \le V \le \varepsilon_0$):}
Now $\displaystyle
\frac{V^{\gamma}}{\varepsilon_0^{\gamma}+V^{\gamma}}
\ge \frac{V^{\gamma}}{2\varepsilon_0^{\gamma}}$,
hence, removing the constant negative term $-(k_0-\bar d)$ for a conservative bound leads to
\begin{equation}
\dot V \le -\frac{k_1}{2\varepsilon_0^{\gamma}}\,V^{\gamma}
= -\alpha V^{\gamma},
\qquad \alpha =\frac{k_1}{2\varepsilon_0^{\gamma}}.
\end{equation}
Integrating $\dot V \le -\alpha V^{\gamma}$ from $V_{\text{in}}=\min\{|x_0|,\varepsilon_0\}$ down to $V=\varepsilon$
yields
\begin{align}
    \begin{split}
 t_2 &\le \frac{1}{\alpha(1-\gamma)}\Big(\varepsilon^{\,1-\gamma}-V_{\text{in}}^{\,1-\gamma}\Big)_+ \\
& = \frac{\varepsilon_0^{\gamma}}{k_1(1-\gamma)}
      \Big(\varepsilon^{\,1-\gamma}-\min\{|x_0|,\varepsilon_0\}^{\,1-\gamma}\Big)_+ .
\label{eq:t2-range2}       
    \end{split}
\end{align}
Adding \eqref{eq:t1-range1} and \eqref{eq:t2-range2} gives the FT boundary layer-entry bound $T_{\mathrm{out}}$ given by \eqref{eq:Tout-bound}.

\medskip
\noindent\textbf{Case 2:} $|x| \le \varepsilon$.\\[0.3em]
With the control law as mixed-power  ~\eqref{eq:Gin-poly} on substituting in  ~\eqref{eq: Vdot first order derived},
leads to
\begin{equation}
\dot V = -aV^{\gamma}-bV^{\alpha}+d(t)\sign(x)
\le -aV^{\gamma}-bV^{\alpha},
\end{equation}
which yields the standard FxT estimate
\begin{equation}
T_{\mathrm{in}} \;\le\; \frac{1}{a(1-\gamma)} + \frac{1}{b(\alpha-1)}.
\label{eq:Tin-mixedpower}
\end{equation}
Therefore, the total settling time satisfies
$T_{\mathrm{tot}} \le T_{\mathrm{out}} + T_{\mathrm{in}}$.
\end{proof}

\begin{corollary}[Exponential inner (erf) variant]
\label{cor:erf-inner}
Replace the polynomial law \eqref{eq:Gin-poly} inside $|x|\le\varepsilon$  by the exponential function form~\cite{labbadi2023design}
\begin{equation}
G_{\mathrm{in}}(x) = \frac{\sqrt{\pi}}{2}\,U\,e^{x^2},
\qquad U>\frac{2}{\sqrt{\pi}}\bar d,
\label{eq:Gin-erf}
\end{equation}
The time inside the boundary layer follows the fixed bound
\begin{equation}
T_{\mathrm{in}}^{\mathrm{erf}}
\le \frac{\varepsilon}{\frac{\sqrt{\pi}}{2}U-\bar d},
\label{eq:Tin-erf}
\end{equation}
so that $T_{\mathrm{tot}} \le T_{\mathrm{out}} + T_{\mathrm{in}}^{\mathrm{erf}}$.
\end{corollary}


\begin{remark}[Bounded effort and tuning]
In the outer region, $k_0<G_{\mathrm{out}}(x)<k_0+k_1$, so the control is uniformly bounded even for large $|x|$.
Choose $k_0=\bar d+\Delta$ ($\Delta>0$) to set the robustness margin; then $(k_1,\varepsilon_0,\gamma)$ shape the FT transient and the weight of the nonlinear term.
Inside the boundary layer, either the \emph{polynomial} law \eqref{eq:Gin-poly} gives the tunable FxT constant \eqref{eq:Tin-poly}, 
or the \emph{erf} law \eqref{eq:Gin-erf} yields the tight linear bound \eqref{eq:Tin-erf}, with the exponential factor automatically capped by $e^{\varepsilon^2}$ due to the region switch.
\end{remark}
\subsection{EL Systems}
\begin{theorem}
\label{thm:EL-hybrid}
For the dynamics~\eqref{eq:EL}, the proposed control law
\begin{equation}
\tau = M(q)\ddot q_r + C(q,\dot q)\dot q_r + g(q)
       - K_{\mathrm{hyb}}(s)\,\mathrm{sgn}(s)
\label{eq:EL-control-final}
\end{equation}
drives the tracking error $e = q_r-q_d$ and $\dot e = \dot {q}_r-\dot q_d$ to zero, which eventually leads to sliding surface $s = \dot e + \Lambda e$  to zero in finite time estimated as \eqref{eq: outer region} and \eqref{eq: inner region}.
Here, $\Lambda = \mathrm{diag}(\lambda_1,\dots,\lambda_n),\;\lambda_i>0$. The hybrid gain is
\begin{align}
   K_{\mathrm{hyb}}(s)=\mathrm{diag}(K_1(s_1),\dots,K_n(s_n))
   \label{eq: hybrid gain}
\end{align}
with outer and inner region formulation is considered as
\begin{equation}
K_i(s_i)=
\begin{cases}
\displaystyle k_{0i} + k_{1i}\frac{|s_i|^\gamma}{\varepsilon_0^\gamma+|s_i|^\gamma},
& |s_i|>\varepsilon,\\[1.5ex]
\displaystyle a_i|s_i|^\gamma + b_i|s_i|^\alpha,
& |s_i|\le\varepsilon,
\end{cases}
\label{eq:EL-hybrid-gain}
\end{equation}
and parameters satisfying
\begin{align}
 \begin{split}
     k_{0i}&>\bar d_i, \ k_{1i}>0,\quad a_i>0,\quad b_i>0, \\
\varepsilon_0&>0,\quad 0<\gamma<1<\alpha.
 \end{split}   
\end{align}

ensures the following time bound.

\begin{enumerate}
\item[(i)] \textit{FT  entry to boundary layer.}  
Each component $s_i(t)$ reaches the boundary layer $|s_i|\le\varepsilon$ within finite time bounded by
\begin{align}
T_{\mathrm{out}}^{(i)}
&\le\frac{1}{k_{0i}-\bar d_i}
   \ln\!\frac{\max\{|s_i(0)|,\varepsilon_0\}}{\max\{\varepsilon,\varepsilon_0\}}
\notag\\
&\quad+\frac{\varepsilon_0^\gamma}{k_{1i}(1-\gamma)}
\Big(\varepsilon^{1-\gamma}-\min\{|s_i(0)|,\varepsilon_0\}^{1-\gamma}\Big)_+ .
\label{eq: outer region}
\end{align}
Thus $T_{\mathrm{out}}=\max_i T_{\mathrm{out}}^{(i)}<\infty$.

\item[(ii)] \textit{FxT within the boundary layer.}  
Once $\|s\|_\infty\le\varepsilon$, the sliding variable converges to the origin within a uniform,
state-independent bound
\begin{equation}
T_{\mathrm{in}}
\le \max_{i}
\left\{\frac{1}{a_i(1-\gamma)} + \frac{1}{b_i(\alpha-1)}\right\}.
\label{eq: inner region}
\end{equation}
Hence, the total settling time satisfies
\begin{align}
  T_{\mathrm{tot}} \le T_{\mathrm{out}} + T_{\mathrm{in}}.  
\end{align}
\end{enumerate}
\end{theorem}

\begin{proof}
Using time derivative of $s$ and on
substituting the control law \eqref{eq:EL-control-final}
into \eqref{eq:EL} yields the closed-loop sliding dynamics
\[
M(q)\dot s + C(q,\dot q)s = -K_{\mathrm{hyb}}(s)\,\mathrm{sgn}(s) + d(t).
\]
Consider the Lyapunov function
\[
V(s)=\tfrac12 s^\top M(q)s.
\]
Using the EL identity $\dot M - 2C$ skew-symmetric,
\[
\dot V = s^\top(M\dot s + \tfrac12\dot M s)
       = s^\top(-K_{\mathrm{hyb}}(s)\mathrm{sgn}(s) + d(t)).
\]
Thus,
\[
\dot V \le \sum_{i=1}^n (-K_i(s_i)+\bar d_i)|s_i|.
\]

\noindent\textbf{Case 1: $|s_i|>\varepsilon$.}\;
Substituting the outer region gain shows $|s_i|$ strictly decreases.
Splitting the trajectory into the ranges
$|s_i|\ge\varepsilon_0$ and $\varepsilon\le|s_i|\le\varepsilon_0$
and integrating yields $T_{\mathrm{out}}^{(i)}$.

\noindent\textbf{Case 2: $|s_i|\le\varepsilon$.}\;
The inner region gain yields the mixed-power inequality
\[
\dot{|s_i|} \le -a_i|s_i|^\gamma - b_i|s_i|^\alpha,
\]
which satisfies Polyakov’s FxT lemma,
giving the uniform bound $T_{\mathrm{in}}$.

Combining both phases proves the result $T_{\mathrm{tot}} \le T_{\mathrm{out}} + T_{\mathrm{in}}^{\mathrm{erf}}.$.
\end{proof}

\begin{corollary}[Exponential Inner Law]
If the polynomial law within the boundary layer is replaced by
\[
K_i(s_i)=\tfrac{\sqrt\pi}{2}U_i e^{s_i^2},
\qquad
U_i>\tfrac{2}{\sqrt\pi}\bar d_i,
\]
then inside the boundary layer the convergence time satisfies
\[
T_{\mathrm{in}}^{\mathrm{erf}}
\le \max_i\frac{1}{U_i-\frac{2}{\sqrt\pi}\bar d_i},
\]
and since $|s_i|\le\varepsilon$ implies $e^{s_i^2}\le e^{\varepsilon^2}$,
the control remains bounded. Hence
\[
T_{\mathrm{tot}} \le T_{\mathrm{out}} + T_{\mathrm{in}}^{\mathrm{erf}}.
\]
\end{corollary}

\begin{remark}[Boundedness and Tuning Guidelines]
For $|s_i|>\varepsilon$, the outer gain satisfies
\[
k_{0i} < K_i(s_i) < k_{0i}+k_{1i},
\]
ensuring uniformly bounded control away from the origin.
Choose $k_{0i}=\bar d_i+\Delta_i$ to set the robustness floor.
The parameters $(k_{1i},\gamma,\varepsilon_0)$ shape the mid-range decay speed.
Inside the boundary layer, $(a_i,b_i,\gamma,\alpha)$ determine the state-independent
FxT constant $T_{\mathrm{in}}$.
\end{remark}

\section{Simulation Results}\label{sec: simulation results}
This section presents numerical simulation studies to evaluate the performance of the proposed controller design for the two system dynamics introduced in~\eqref{eq:first-order} and~\eqref{eq:EL}. 
The analysis is first carried out using the first-order integrator model and,
subsequently, the controller is implemented for a trajectory tracking scenario to demonstrate its practical applicability to EL systems.

\subsection{First-order system}
The first set of simulations considers the uncertain first-order integrator system described by~\eqref{eq:first-order}. 
The disturbance, simulation time step, and other numerical settings are provided in Table~\ref{tab:sim-settings-1d}. 
Controller parameters for the proposed Hybrid--Poly and Hybrid--Erf methods, as well as the benchmark SATO-based approach~\cite{li2021simultaneous}, are summarized in Table~\ref{tab:ctrl params compact}. 
All parameters were tuned such that the three controllers yield comparable settling times, thereby ensuring a fair basis for performance comparison.

The tracking performance is evaluated using three quantitative metrics: 
(i) the root mean square (RMS) error, 
\(
\mathrm{RMS} = \sqrt{\frac{1}{T_{\mathrm{tot}}}\int_{0}^{T_{\mathrm{tot}}}\|x(t)\|_2^2\,dt},
\)
(ii) the integral of absolute error (IAE),
\(
\mathrm{IAE} = \int_{0}^{T_{\mathrm{tot}}}\sum_{i=1}^{n}|x_i(t)|\,dt,
\)
and (iii) the average control magnitude,
\(
\overline{\|u\|} = \frac{1}{T_{\mathrm{tot}}}\int_{0}^{T_{\mathrm{tot}}}\sqrt{\sum_{i=1}^{m}u_i^2(t)}\,dt.
\)
The corresponding values are reported in Table~\ref{tab: comparison value}, while representative state and control trajectories are shown in Fig.~\ref{fig: Results for comparative controller design}. The state trajectory is plotted in Fig.~\ref{fig: x1 first order} and the corresponding control profile is shown in Fig.~\ref{fig: control law first order}.
\begin{remark}[Comparative Performance Analysis]
 From Figs.~\ref{fig: x1 first order} and~\ref{fig: control law first order}, it can be observed that both the proposed Hybrid--Poly and Hybrid--Erf controllers attain equilibrium with lower chattering amplitudes compared to the SATO-based method, while exhibiting negligible differences in the overall settling time. 
However, as seen from Table~\ref{tab: comparison value}, the proposed controllers yield improved tracking accuracy (RMS error reduction of approximately $1.9\%$ and IAE reduction of $2.6\%$) and, more importantly, achieve a significant reduction in the average control effort---about $31\%$ compared to the SATO baseline. 
This improvement confirms that the hybrid-gain structure enables smoother and more energy-efficient control action without compromising robustness or convergence speed.   
\end{remark}

To further validate robustness, the proposed controllers were tested under multiple initial conditions, $x(0) = [\,4,\,-5,\,6\,]^{\mathrm{T}}$. 
The resulting state and control trajectories for the polynomial and exponential inner laws are illustrated in Figs.~\ref{fig: Results for proposed method with polynomial inner law} and~\ref{fig: Results for proposed method with exponential inner law}, respectively. 
In both cases, the proposed controllers maintain consistent convergence characteristics, and even for smaller initial errors (e.g., $x(0)=4$), the SATO-based control exhibits noticeably higher actuation levels than the hybrid-gain based design. 
These results further reinforce that the proposed HG--FTSMC ensures uniform FT convergence while substantially reducing control effort.
\subsection{ EL System}
In this case, 
numerical simulations are conducted on a two-link planar manipulator modeled as
a representative EL system described by~\eqref{eq:EL}.
Here $q=[\,q_1,\,q_2\,]^{\mathrm{T}}$ denotes the joint angles of the first and second links, respectively,
and $\dot q=[\,\dot q_1,\,\dot q_2\,]^{\mathrm{T}}$ are their angular velocities. 
Please refer~\cite{kumari2018event} for detailed modeling and simulation parameters.


The goal is to ensure accurate tracking of a desired joint trajectory
\begin{align}
  q_d(t) = [\,0.1\sin(\pi t),\; 0.2\sin(\pi t)\,]^{\mathrm{T}}.  
\end{align}
in the presence of matched external disturbances 
\begin{align} 
  d(t) = [\,1.0\sin t,\; 0.1\cos t\,]^{\mathrm{T}}.  
\end{align}
The inertia matrix is assumed to be bounded such that $\mu_1 I \le M(q) \le \mu_2 I$
with $\mu_1=1$ and $\mu_2=2$. The system parameters and constants are selected as
\begin{align}
\begin{split}
  p_1 &= 3.473~\text{kg\,m}^2, \quad 
p_2 = 0.196~\text{kg\,m}^2, \ 
 \\
p_3 &= 0.242~\text{kg\,m}^2, \ f_{d1}=f_{d2}=1.1~\text{N\,m\,s}
\end{split}
\end{align}
The sliding variable is defined as
\begin{equation}
s = \dot e + C e,
\qquad
C =
\begin{bmatrix}
0.5 & 0\\
0 & 0.5
\end{bmatrix},
\label{eq:EL-sliding}
\end{equation}
where $e=q-q_d$ and $\dot e=\dot q-\dot q_d$.
The hybrid-gain control law $\tau$ follows the form in~\eqref{eq:EL-control-final},
with $K_\mathrm{hyb}(s_i)$ designed according to~\eqref{eq:EL-hybrid-gain}.
The design constants are chosen as
$\alpha=0.2$, $\sigma=0.85$, $L=3$, and the initial conditions are
$q(0)=[\,2,\,2\,]^{\mathrm{T}}$ and $\dot q(0)=[\,0,\,0\,]^{\mathrm{T}}$.
The controller gains are tuned as
$K_1=2.58$, $K_2=1.4$, and $\beta=3.676$. 

Simulation results 
are illustrated in Fig.~\ref{fig: Results for two link manipulator}.
Figure~\ref{fig: States trajectory} shows that both joint angles $q_1(t)$ and $q_2(t)$
accurately track their respective desired trajectories $q_{d1}(t)$ and $q_{d2}(t)$.
The corresponding control torques $u_1$ and $u_2$, shown in Fig.~\ref{fig: Control profile}
The Hybrid--Erf controller produces marginally smoother torque profiles
due to the exponential inner gain term.
The error profiles for $e_1$, $e_2$, and the sliding variable norm $\|s\|$
is depicted in Fig.~\ref{fig: error profiles}.
Both controllers achieve rapid FT convergence of the joint tracking errors
to a small neighborhood around zero and maintain this bound thereafter,
demonstrating robustness to model uncertainties and external disturbances.
\begin{table}
\centering
\caption{Simulation settings for first-order integrator model}
\label{tab:sim-settings-1d}
\renewcommand{\arraystretch}{1.1}
\begin{tabular}{ll}
\toprule
Description & Value \\
\midrule
Disturbance & $d(t)=0.8\,\bar d \sin(5t)$, \quad $\bar d = 0.5$ \\
Step size/Run time & $dt=10^{-3}\ \mathrm{s}$,\; $T=6\ \mathrm{s}$ \\
boundary layer radius & $\varepsilon = 0.08$ \\
Final settle radius & $\delta = 10^{-3}$ \\
Initial condition & $x_0 = 3$ \\
\bottomrule
\end{tabular}
\end{table}

\begin{table}
\centering
\caption{Controller parameters for the three methods}
\label{tab:ctrl params compact}
\renewcommand{\arraystretch}{1.15}
\begin{tabular}{lcc}
\toprule
\textbf{Method} & \textbf{Parameters } & \textbf{Values} \\
\midrule
\multirow{2}{*}{Hybrid--Poly} 
& $k_0, k_1, \varepsilon_0, \gamma $
& $0.8,0.8,0.25,0.7$ \\
& $a, b, \alpha, \varepsilon$& $2.5,1.2,1.8,0.08$ \\
\multirow{2}{*}{Hybrid--Erf}
& $k_0, k_1, \varepsilon_0$
& $0.8,0.8,0.25$ \\
&$\gamma, c_{\mathrm{erf}}, \varepsilon$ & $0.7,1.2,0.08$\\
SATO  
& $K_{\mathrm{nn}}, \sigma, \varepsilon$
& $2.8\,\bar{d},\,10^{-9},\,0.08$ \\
\bottomrule
\end{tabular}
\end{table}

\begin{figure}
    \begin{subfigure}[b]{.25\textwidth}
			\centering			\includegraphics[width=\linewidth]{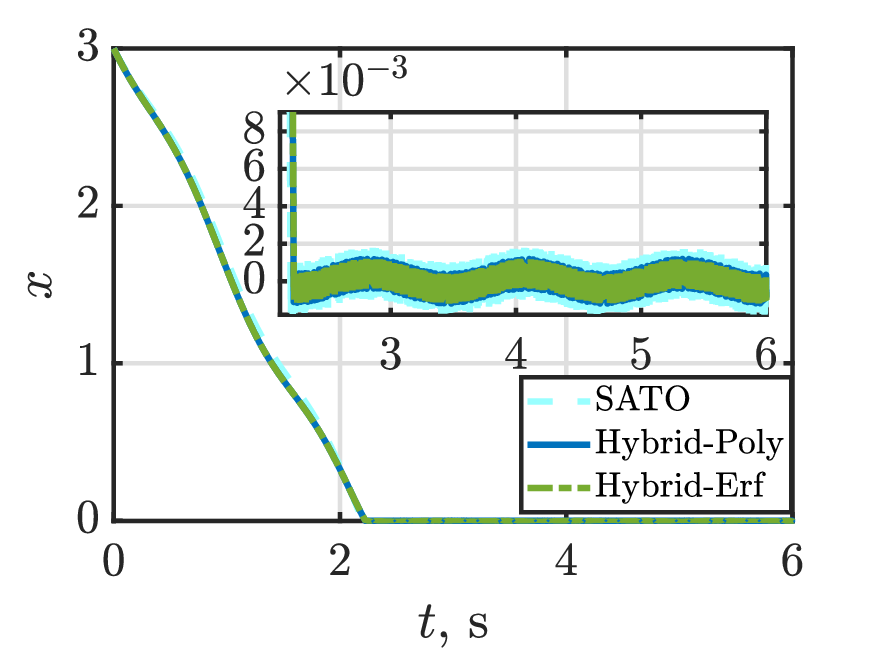}
			\caption{ State trajectory}
			\label{fig: x1 first order}     
		\end{subfigure}%
		\begin{subfigure}[b]{.25\textwidth}
			\centering		\includegraphics[width=\linewidth]{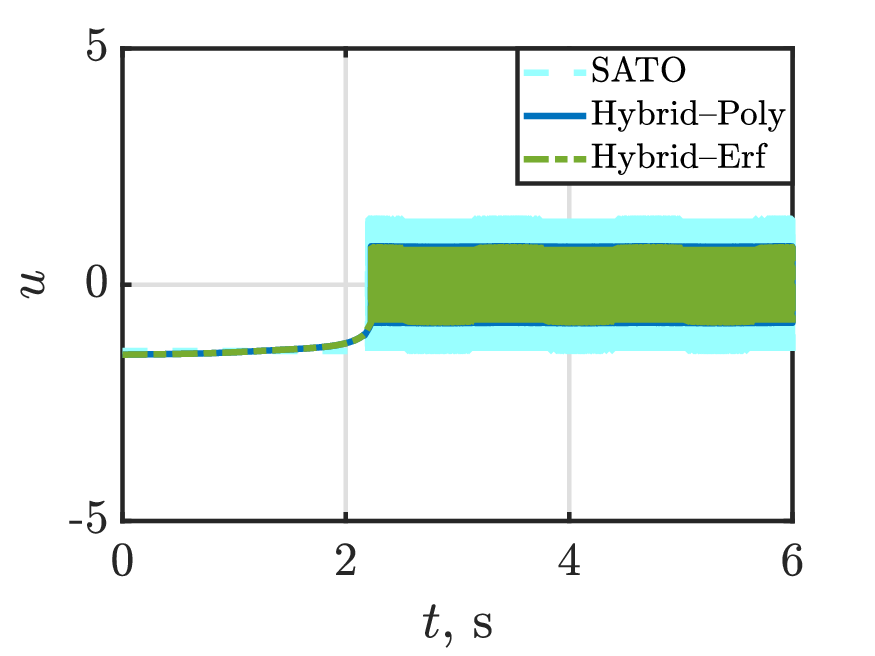}
			\caption{ Control profile }
			\label{fig: control law first order}
		\end{subfigure}%
         \caption{Results for comparative controller design}
			\label{fig: Results for comparative controller design}
\end{figure}
\begin{figure}
    \begin{subfigure}[b]{.25\textwidth}
			\centering			\includegraphics[width=\linewidth]{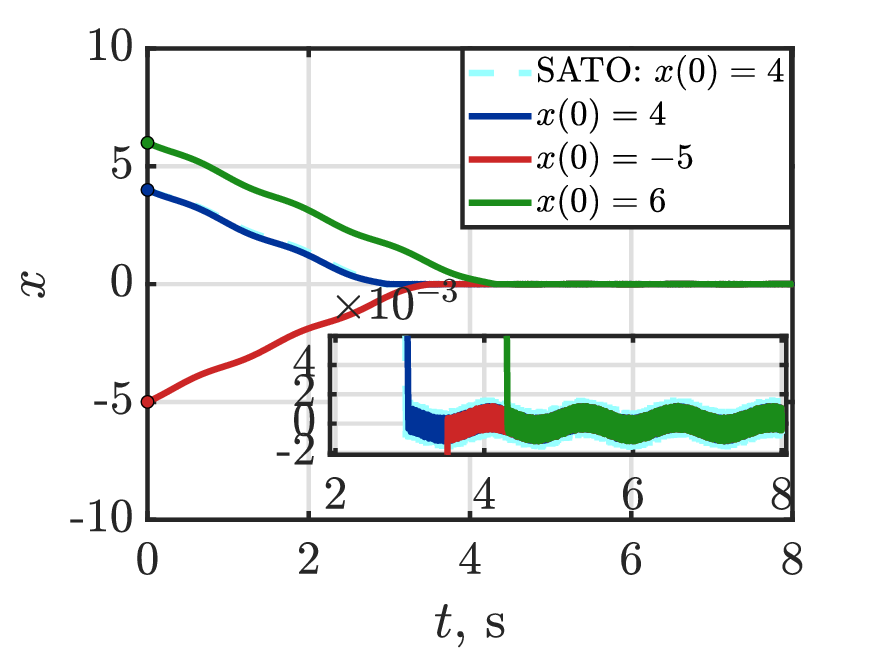}
			\caption{ State trajectory}
			\label{fig: x1 first order poly sota}     
		\end{subfigure}%
		\begin{subfigure}[b]{.25\textwidth}
			\centering		\includegraphics[width=\linewidth]{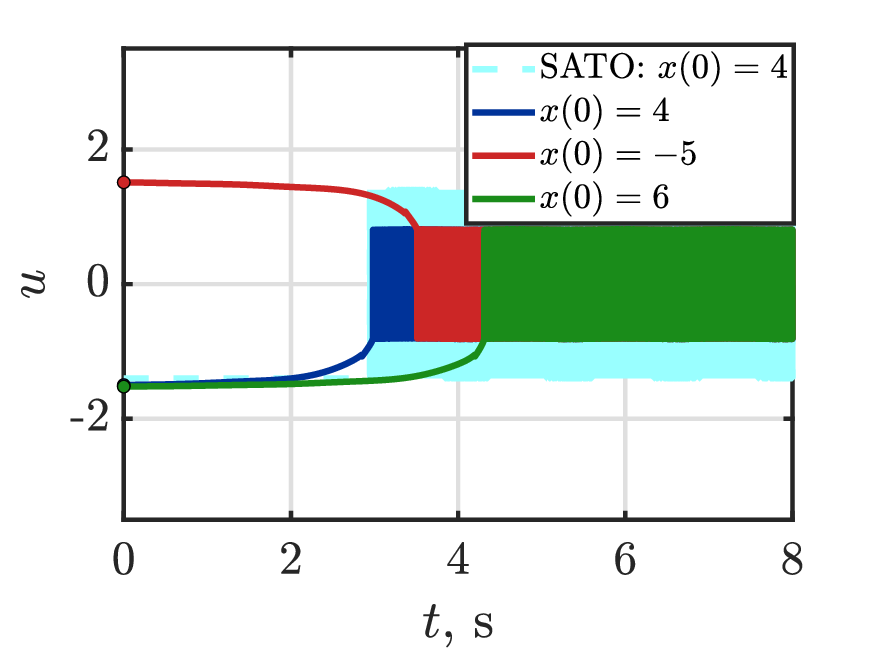}
			\caption{ Control profile }
			\label{fig: control law first order poly sota}
		\end{subfigure}%
         \caption{ Results for HG-FTSMC with polynomial inner law }
			\label{fig: Results for proposed method with polynomial inner law}
\end{figure}
\begin{figure}
    \begin{subfigure}[b]{.25\textwidth}
			\centering		\includegraphics[width=\linewidth]{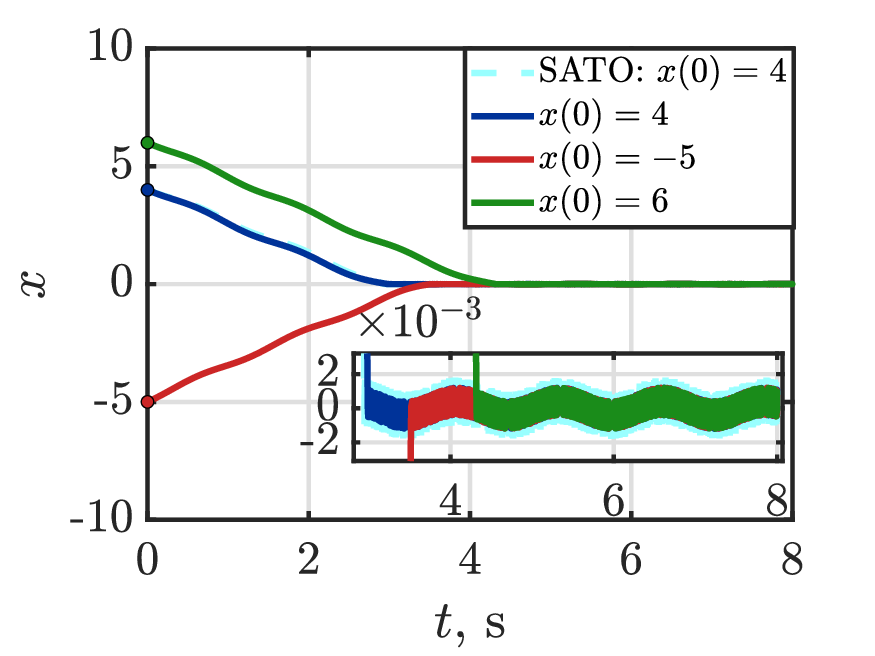}
			\caption{ State trajectory}
			\label{fig: x1 first order erf sota}     
		\end{subfigure}%
		\begin{subfigure}[b]{.25\textwidth}
			\centering		\includegraphics[width=\linewidth]{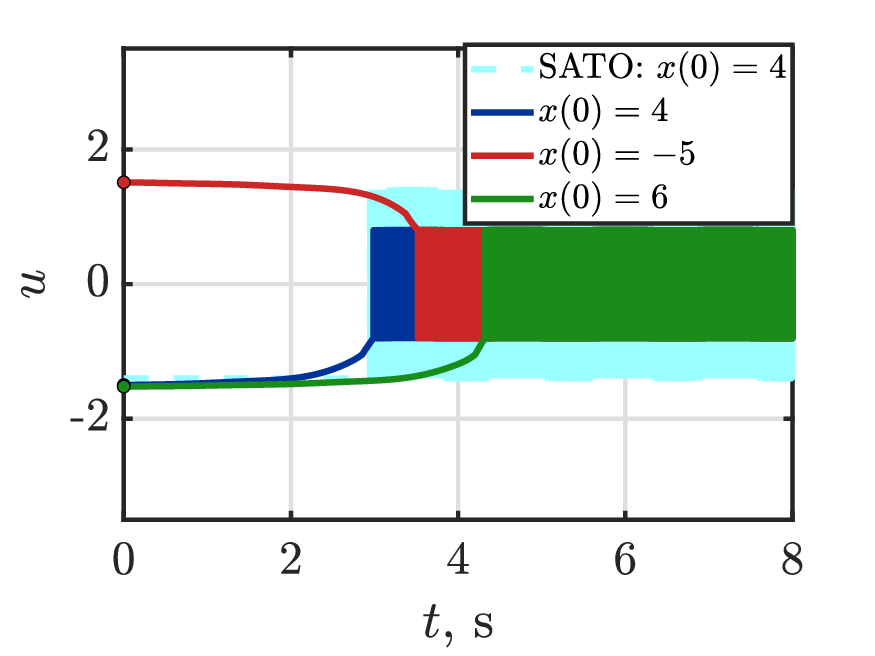}
			\caption{ Control profile }
			\label{fig: control law first order erf sota}
		\end{subfigure}%
        \caption{ Results for HG-FTSMC with exponential inner law }
			\label{fig: Results for proposed method with exponential inner law}
\end{figure}
\vspace{-15pt}
\begin{table}
\centering
\caption{Performance comparison with ($x_0 = 3$).}
\label{tab: comparison value}
\renewcommand{\arraystretch}{1.15}
\begin{tabular}{lcccc}
\toprule
\textbf{Metric} & \textbf{Hybrid-Poly} & \textbf{SATO } & \textbf{Hybrid-Erf}& \textbf{$\%$ decrease} \\
\midrule
$T_{\mathrm{tot}}$ & 7.994 &  8.000 & 7.995 &   $ \longdash$ \\
$x_{\mathrm{rms}}$                     & 0.9145 & 0.9318& 0.9145& \textbf{1.9}\\
$\mathrm{IAE} $          & 3.3160 & 3.4052 & 3.3161& \textbf{2.6} \\
$\|u\| $ & 0.9663 & 1.4000 & 0.9621 &  \textbf{31 }\\
\bottomrule
\end{tabular}
\end{table}

\begin{figure}
		\begin{subfigure}[b]{.25\textwidth}
			\centering			\includegraphics[width=\linewidth]{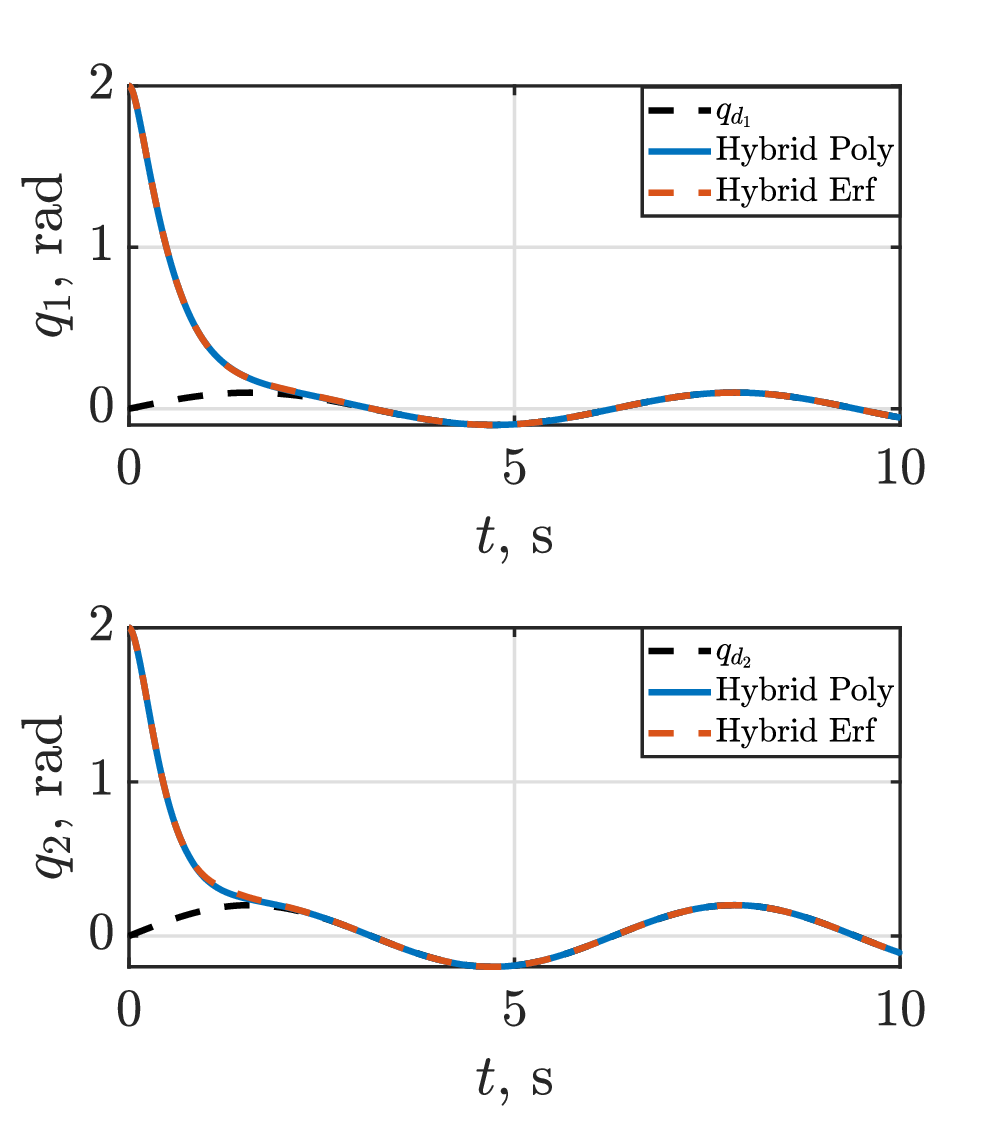}
			\caption{ States trajectory}
			\label{fig: States trajectory}
		\end{subfigure}%
		\begin{subfigure}[b]{.25\textwidth}
			\centering			\includegraphics[width=\linewidth]{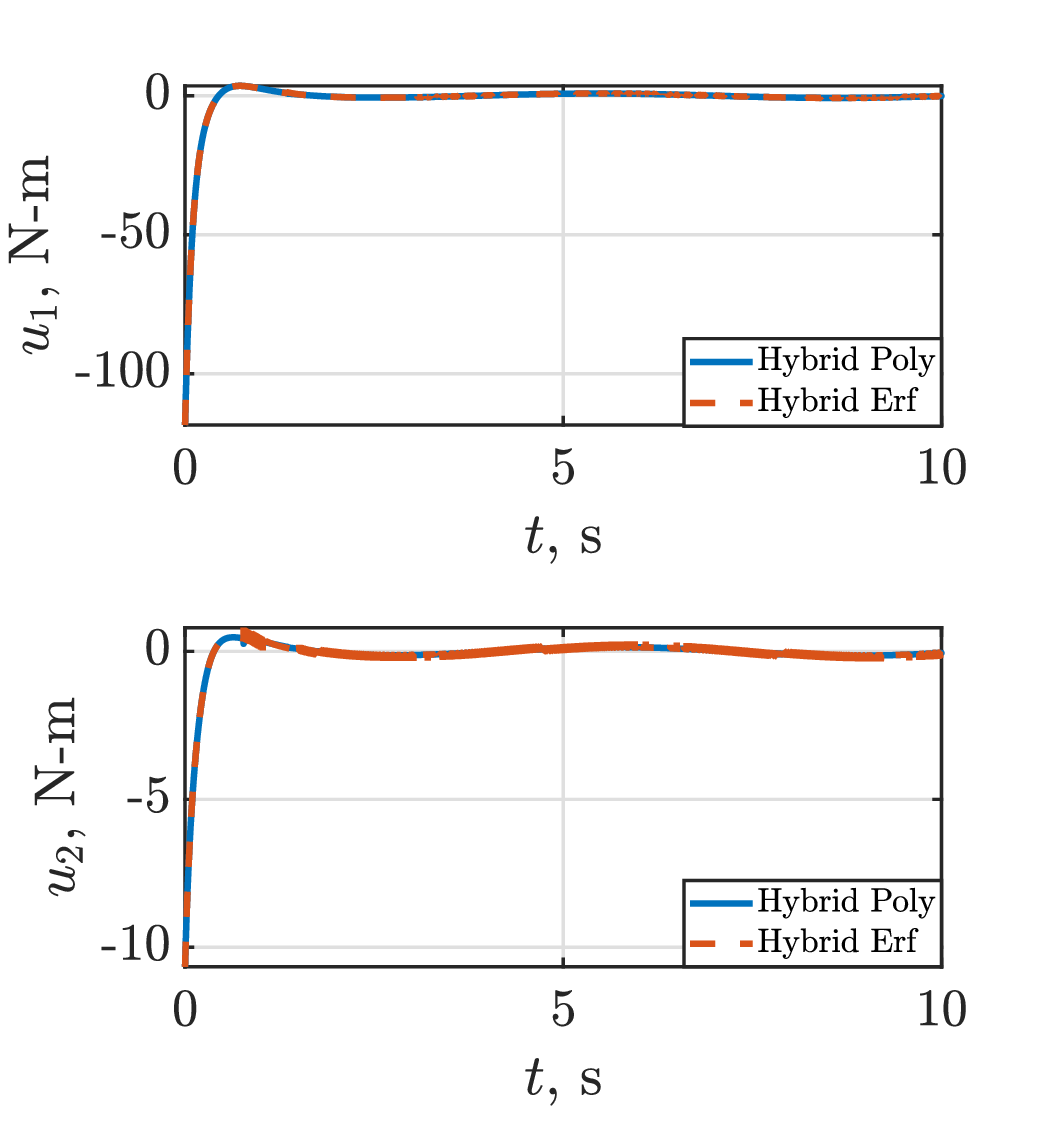}
			\caption{ Control profiles }
			\label{fig: Control profile}
		\end{subfigure}%
        \qquad
        \centering
        \begin{subfigure}[b]{.25\textwidth}			\includegraphics[width=1.4\linewidth]{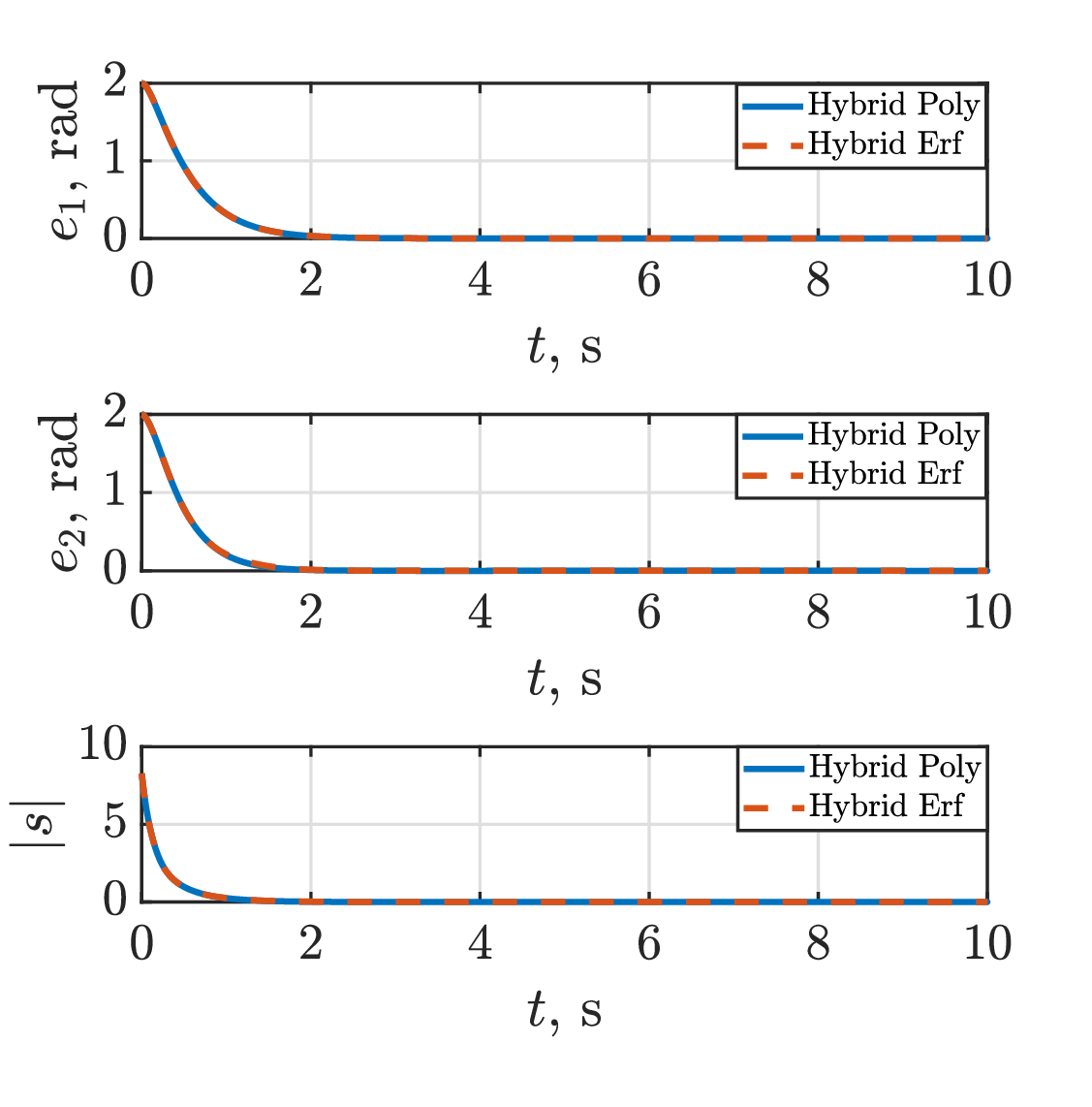}
			\caption{Error profiles}
			\label{fig: error profiles}
		\end{subfigure}%
		\caption{Results for two link manipulator}
		\label{fig: Results for two link manipulator}
	\end{figure}


\section{Conclusion}
\label{sec:conclusion}

This work presented a robust FT sliding-mode control approach based on a hybrid gain structure that achieves rapid convergence with reduced control effort.
By combining a bounded FT law in the outer region with a mixed-power FxT term in the inner region, the HG-FTSMC provides FT reaching, and strong robustness against matched disturbances without inducing excessive control peaks.
Simulations on a perturbed first-order system illustrate that both Hybrid-Poly and Hybrid-Erf variants attain settling times comparable to the established SATO method while reducing average control effort by more than 30$\%$ and offering marginal improvements in tracking accuracy.
A two-link manipulator study further demonstrated the practicality of the proposed framework, yielding smooth and precise trajectory tracking performances.
Overall, the hybrid-gain strategy offers an easily implementable solution for fast, efficient, and robust stabilization of Euler–Lagrange systems. Future research directions include extending the finite-/FxT control framework to systems subject to disturbances with unknown or time-varying upper bounds.
\vspace{-15pt}
\bibliographystyle{ieeetr}
\bibliography{refICC2025.bib}

\end{document}